\documentclass[letterpaper, 10 pt, conference]{ieeeconf}

\usepackage{graphicx}
\usepackage{framed}
\usepackage{subfig}
\usepackage{indentfirst}
\usepackage[]{algorithm2e}
\usepackage{booktabs}

\usepackage{amsmath, amssymb, amsthm}
\usepackage{multirow} 
\usepackage{cleveref}
\usepackage[keeplastbox]{flushend}
\usepackage{microtype}
\usepackage{siunitx}
\crefname{figure}{Fig.}{Fig.}
\Crefname{figure}{Fig.}{Fig.}

\usepackage{pgfplots}
\pgfplotsset{compat=newest}
\usetikzlibrary{plotmarks}
\usetikzlibrary{arrows.meta}
\usepgfplotslibrary{patchplots}
\usepackage{grffile}
\usepackage{amsmath}

\renewcommand{\vec}[1]{\boldsymbol{\mathrm{#1}}} 

\let\labelindent\relax
\usepackage{enumitem}

\usepackage[font = small]{caption}

\newtheorem{prop}{Proposition}

\IEEEoverridecommandlockouts                              

\title{\LARGE \bf Analytically Modeling Unmanaged Intersections \\ with Microscopic Vehicle Interactions}
\author{Changliu Liu and Mykel J. Kochenderfer
\thanks{C. Liu and M. Kochenderfer are with the Department of Aeronautics and Astronautics, Stanford University, CA 94305 USA (e-mail: \tt\small changliuliu, mykel@stanford.edu).}%
}

\begin{document}

\maketitle

\begin{abstract}
With the emergence of autonomous vehicles, it is important to understand their impact on the transportation system. However, conventional traffic simulations are time-consuming. This paper introduces an analytical traffic model for unmanaged intersections accounting for microscopic vehicle interactions. The macroscopic property, i.e., delay at the intersection, is modeled as an event-driven stochastic dynamic process, whose dynamics encode the microscopic vehicle behaviors. The distribution of macroscopic properties can be obtained through either direct analysis or event-driven simulation. They are more efficient than conventional (time-driven) traffic simulation, and capture more microscopic details compared to conventional macroscopic flow models. We illustrate the efficiency of this method by delay analyses under two different policies at a two-lane intersection. The proposed model allows for 1) efficient and effective comparison among different policies, 2) policy optimization, 3) traffic prediction, and 4) system optimization (e.g., infrastructure and protocol).
\end{abstract}

\section{Introduction}
With the emergence of autonomous vehicles, it is important to understand how the microscopic interactions of those autonomous vehicles affect the delay of the macroscopic traffic flow, especially at unmanaged intersections. 

The literature contains many traffic models that can support the analysis of delay and congestion \cite{hoogendoorn2001state}. There are two major types of traffic models: 1) microscopic simulation models where every car is tracked and 2) macroscopic flow models where traffic is described by relations among aggregated values such as flow speed and density, without distinguishing its constituent parts. The major advantage of microscopic simulation models is the precise description of inter-vehicle interactions. Such models have been widely adopted in evaluating the performance of autonomous vehicles \cite{gora2016traffic}. 
However, it can be time-consuming to obtain the micro-macro relationships by simulation. Only ``point-wise" evaluation can be performed in the sense that a single parametric change in vehicle behavior requires new simulations. To gain a deeper understanding of the micro-macro relationships, an analytical model is desirable. 

Macroscopic flow models provide a tractable mathematical structure with few parameters to describe interactions among vehicles. Those models usually come in the form of partial differential equations. However, it remains challenging to model intersections. Existing methods introduce boundary constraints to represent intersections \cite{CORTHOUT2012343, flotterod2011operational}. However, these models can only tolerate simple first-in-first-out (FIFO) policies at intersections. To consider other kinds of policies, the vehicles need to be treated as particles that interact among one another, which has not been captured by existing flow models. 

\begin{figure}[t]
\vspace{-10pt}
\begin{center}
\subfloat[\label{fig: intersection a}]{
\includegraphics[width=3cm]{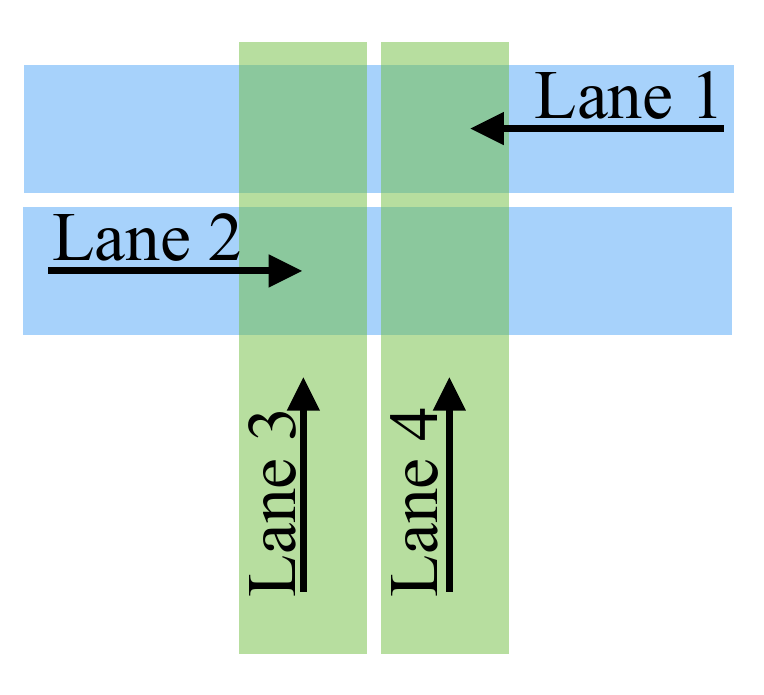}}
\subfloat[\label{fig: intersection b}]{
\includegraphics[width=5cm]{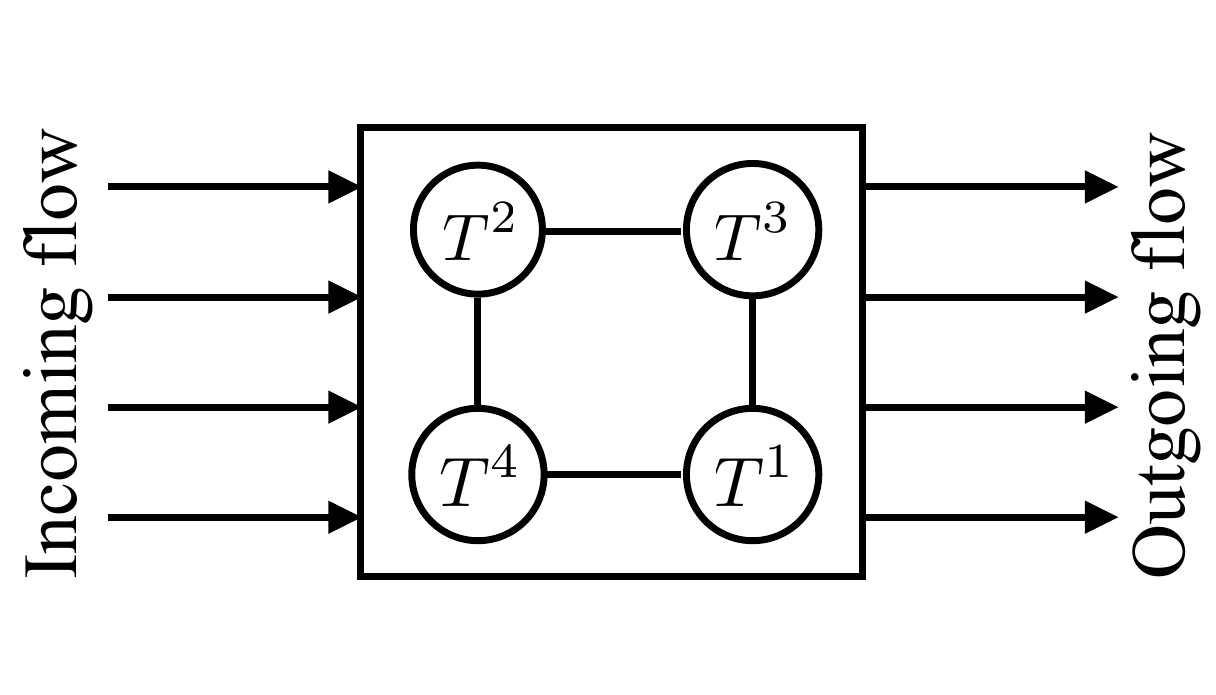}}
\caption{Intersection scenario. (a) Road topology. (b) Conflict graph.}
\label{fig: intersection}
\vspace{-10pt}
\end{center}
\end{figure}

\begin{figure}[t]
\begin{center}
\subfloat[\label{fig: time a}]{
\includegraphics[width=2.7cm]{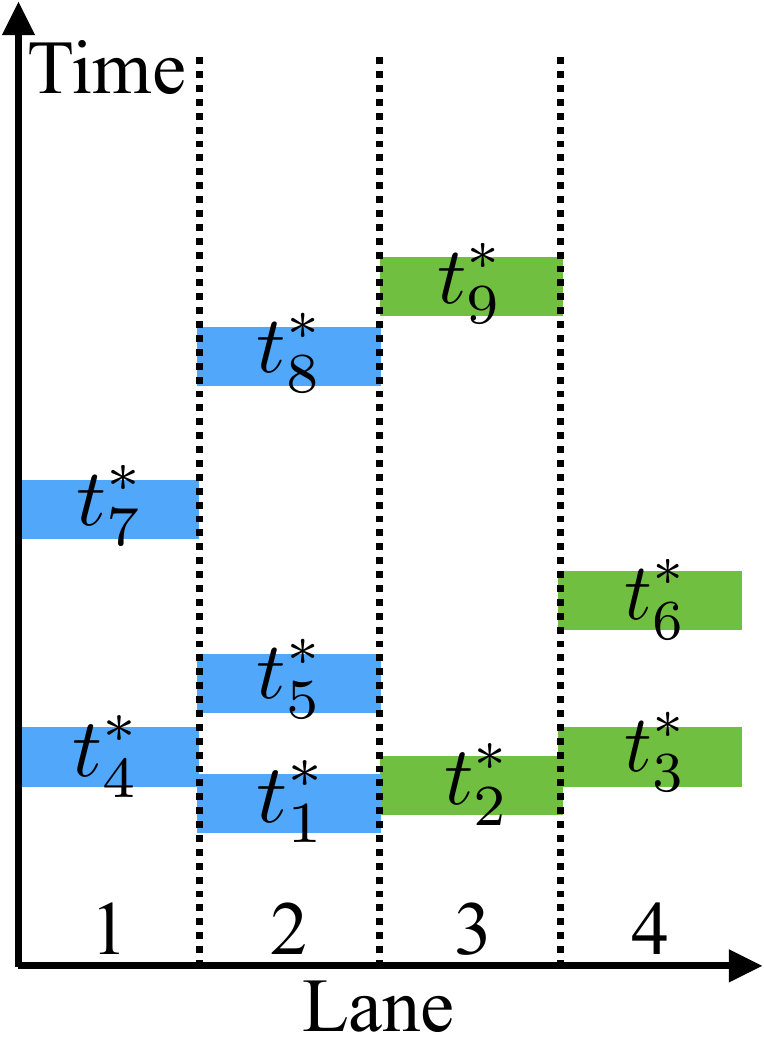}}
\subfloat[\label{fig: time b}]{
\includegraphics[width=2.7cm]{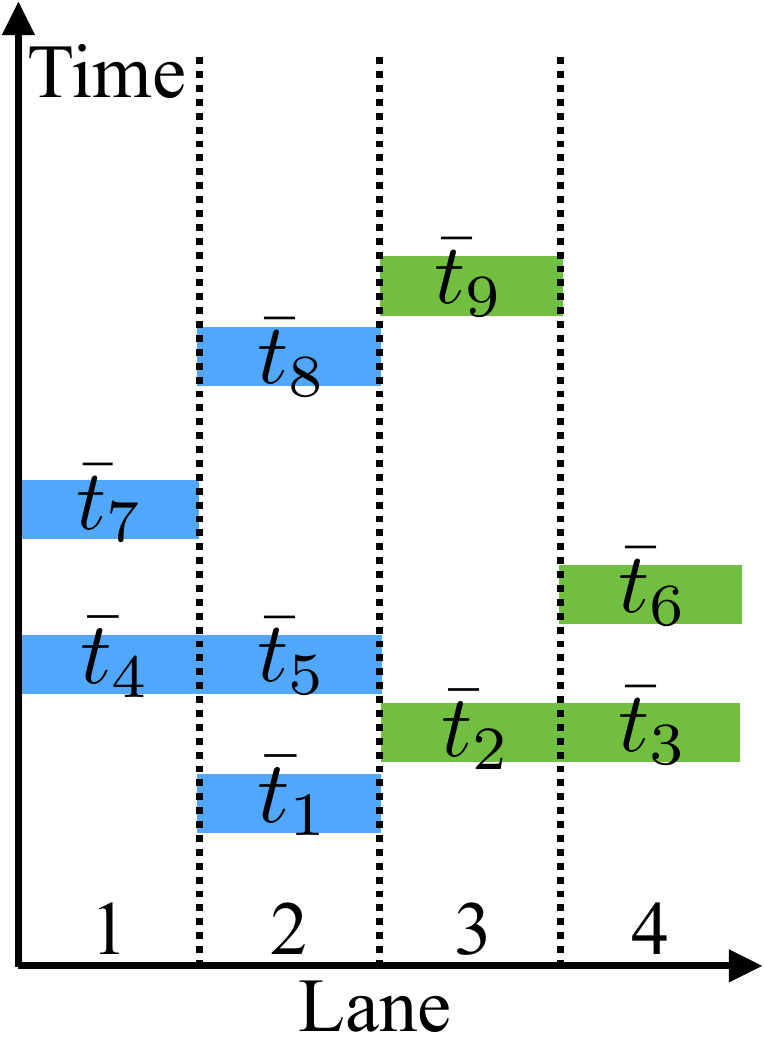}}
\subfloat[\label{fig: time c}]{
\includegraphics[width=2.7cm]{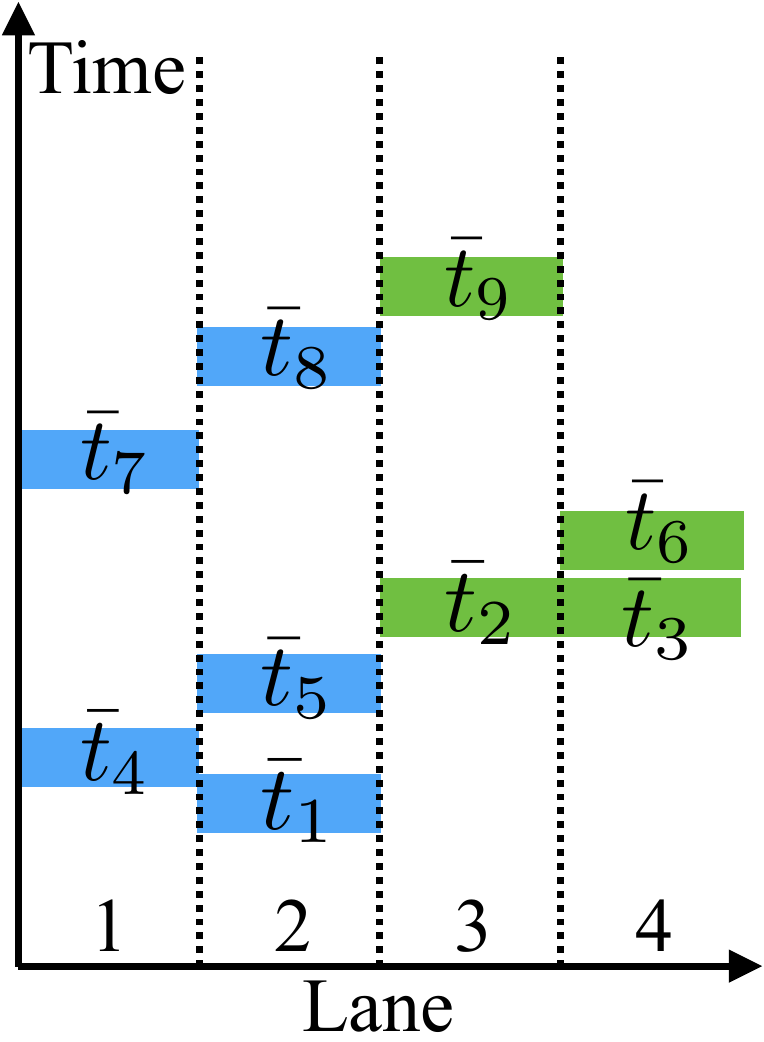}}
\caption{The time of occupancy at the intersection. (a) The desired time of occupancy. (b) The actual time of occupancy under the FIFO policy. (c) The actual time of occupancy under the FO policy.}
\label{fig: time}
\vspace{-10pt}
\end{center}
\end{figure}

This paper introduces an analytical stochastic continuous-time, discrete-event traffic model. We use it to describe delays at unmanaged intersections under different microscopic vehicle behaviors. The model considers microscopic interactions and is analytical, which absorbs the advantages of both the microscopic simulation models and macroscopic flow models. With this model, we can better understand how the microscopic behavior design of a single vehicle affects the macroscopic transportation system. In addition to direct analysis, we can perform event-driven simulation under this model, which is more efficient than conventional time-driven traffic simulation. The policies under consideration are not required to have closed-form solutions. 

The remainder of the paper is organized as follows. \Cref{sec: formulation} formulates the traffic model as an event-driven stochastic process, and illustrates how vehicle behaviors are encoded in the traffic model. \Cref{sec: policy} illustrates the effectiveness of the model through case studies. \Cref{sec: discussion} discusses applications of the method and concludes the paper.

\section{Traffic Model\label{sec: formulation}}
This section introduces the traffic model. The traffic delay at an intersection is modeled as an event-driven stochastic process. An event is defined as the introduction of a new vehicle. This section then describes how the microscopic vehicle interactions affect the macroscopic dynamics.

\subsection{Traffic Model at Intersections}
Consider an intersection with $K$ incoming lanes. A conflict is identified if two incoming lanes intersect. These relationships can be described in a conflict graph $\mathcal{G}$ with the nodes being the incoming lanes and the links representing conflicts. For example, \cref{fig: intersection a} shows one road configuration with four incoming lanes, and \cref{fig: intersection b} shows the conflict graph. $T^k$ is the delay at lane $k$, which will be introduced in \eqref{eq: defn T}. 

When there are conflicts, vehicles from the corresponding lanes cannot occupy the intersection at the same time. Let $t_i^*$ be the desired time for vehicle $i$ to pass the center of the intersection. The vehicles are numbered such that $t_{i}^*<t_{i+1}^*$. \Cref{fig: time a} shows the desired time of occupancy (centered at $t_i^*$) for vehicles coming in the four lanes in \cref{fig: intersection a}. According to the graph in \cref{fig: intersection b}, the scenario in \cref{fig: time a} is infeasible as vehicles 1 through 4 cannot occupy the intersection at the same time. Based on the FIFO policy, vehicles 2 and 3 yield to vehicle 1. Vehicles 4 and 5 yield to vehicles 2 and 3, and so on. Let $\bar t_i$ be the actual time for vehicle $i$ to pass the center of the intersection. \Cref{fig: time b} shows the actual time of occupancy when all vehicles adopt FIFO. The actual time of occupancy may change when the policy changes, resulting in different traffic delay. For example, \cref{fig: time c} corresponds to another policy that will be introduced in \Cref{sec: policy}. 

The traffic at the intersection is modeled as an event-driven stochastic system with the state being the traffic delay and the input being the incoming traffic. It is assumed that the desired passing time for incoming vehicles from lane $k$ follows a Poisson distribution with parameter $\lambda_k$. The traffic flows from different lanes are independent. Since the combination of multiple Poisson processes is a Poisson distribution \cite{gardiner2009stochastic}, the incoming traffic from all lanes can be described as one Poisson process $\{t_1^*,t_2^*,\ldots\}$ with parameter $\lambda = \sum_k \lambda_k$. The input to the model is chosen to be the random arrival interval between vehicle $i+1$ and $i$, i.e., $x_i = t_{i+1}^*-t_i^*$, and the lane number $s_{i+1}$ for vehicle $i+1$. For all $i$, the probability density for $x_i = x$ is $p_x(x) = \lambda e^{-\lambda x}$. The probability of $s_{i+1} = k$ is $P_s(k) = \frac{\lambda_k}{\lambda}$. The delay for lane $k$ considering $i$ vehicles is denoted $T^k_i$, which captures the difference between the actual passing time and the traffic-free passing time of those vehicles, i.e.,
\begin{equation}
T^k_i = \max_{s_{j}=k, j\leq i} \bar t_j^{(i)} - t_{i}^*\text{,}\label{eq: defn T}
\end{equation}
where $\bar t_j^{(i)}$ denotes the actual passing time for vehicle $j$ only considering the interactions among the first $i$ vehicles, e.g., vehicle $i+1$ has not approached the intersection yet. Here, $\bar t_j^{(i)}$ corresponds to an equilibrium in microscopic vehicle interactions which will be introduced in \eqref{eq: i eq}. It may differ from $\bar t_j^{(k)}$ for $k\neq i$. 

Define $\mathbf{T}_i:=[T^1_i, \ldots, T^K_i]^T$. The dynamics of the traffic delay at the intersection is determined by
\begin{equation}
\mathbf{T}_{i+1} = \mathcal{F}(\mathbf{T}_i, x_i, s_{i+1})\text{,}\label{eq: dynamic}
\end{equation}
where the function $\mathcal{F}$ depends on the policies adopted by the vehicles and the road topology defined by the conflict graph $\mathcal{G}$. 
Given \eqref{eq: dynamic}, the conditional probability density of $\mathbf{T}_{i+1}$ given $\mathbf{T}_i$, $x_i$ and $s_{i+1}$ is
\begin{eqnarray}
p_{\mathbf{T}_{i+1}}(\mathbf{t}\mid\mathbf{T}_i, x_i, s_{i+1}) = \delta(\mathbf{t}= \mathcal{F}(\mathbf{T}_i, x_i, s_{i+1}))\text{,}
\end{eqnarray}
where $\delta(\cdot)$ is the delta function. The probability density is
\begin{eqnarray}
&&p_{\mathbf{T}_{i+1}}(\mathbf{t}) \nonumber\\
&=& \sum_k P_{s}(k) \int_{x} \int_{\vec{\tau}}p_{\mathbf{T}_{i+1}}(\mathbf{t}\mid\vec\tau, x, k)p_{\mathbf{T}_i}(\vec{\tau})d\vec{\tau} p_x(x)dx\nonumber\\
&=& \sum_k P_{s}(k) \int_{\mathcal{F}(\vec\tau, x, k)=\mathbf{t}} \delta(0)p_{\mathbf{T}_i}(\vec\tau)p_{x}(x)d\vec\tau dx\label{eq: probability}\text{,}
\end{eqnarray}
which involves integration over a manifold. The cumulative probability of $\mathbf{T}_i$ is denoted as $P_{\mathbf{T}_i}(\mathbf{t})= \int_{-\infty}^{(t^1)^+}\ldots\int_{-\infty}^{(t^k)^+}p_{\mathbf{T}_i}(\tau^1,\ldots,\tau^k) d\tau^1\ldots d\tau^k$ where $\mathbf{t} = [t^1, \ldots, t^k]$. 
The problems of interest are:
\begin{itemize}
\item Does the sequence $\{p_{\mathbf{T}_{i}}\}_i$ converge in $L_1$-norm?

Divergence corresponds to the formation of congestion, i.e., the case that the expected delay keeps growing.
\item If converged, what is the steady state distribution of $p_{\mathbf{T}}:=\lim_{i\rightarrow \infty}p_{\mathbf{T}_{i}}$?

From the steady state distribution, we may compute the expected delay.
\end{itemize}

These two problems will be considered in the case studies in \Cref{sec: policy}. Moreover, from the distribution of the lane delays, we can compute the scalar delay introduced by the $(i+1)$th vehicle as
\begin{equation}
d_{i+1} = \sum_{j\leq i} \left(\bar t_j^{(i+1)} - \bar t_j^{(i)}\right)+\bar t_{i+1}^{(i+1)} - t_{i+1}^{*}\text{.}\label{eq: delay}
\end{equation}
In the case that the introduction of a new vehicle only affects the last vehicle in other lanes (which is usually the case),
\begin{equation}
d_{i+1} = T_{i+1}^{s_{i+1}} + \sum_{k\neq s_{i+1}} (T_{i+1}^k - T_{i}^k + x_i)\text{.} \label{eq: delay macro} 
\end{equation}

\subsection{Microscopic Interactions\label{sec: vehicle}}
It is assumed that the vehicles at intersections have fixed paths. When interacting with other vehicles, they only change their speed profiles to adjust the time to pass the intersection. Such simplification is widely adopted \cite{altche2016time, qian2017autonomous}. In this paper, we further reduce the high dimensional speed profile for vehicle $i$ to a single state $t_i$ which denotes the time for vehicle $i$ to pass the center of the intersection. Since the mapping from $t_i$ to a speed profile is surjective, interactions can be analyzed using $t_i$'s.

The policy of vehicle $i$ is denoted
\begin{equation}
t_i(k) = f(t_i^*,t_{-i}(k-1))\text{,}
\end{equation}
where $k$ denotes time step. The subscript $-i$ denotes all other indices except $i$. The one step delay is due to reaction time. An equilibrium is achieved if the vehicles do not have incentives to adjust the passing time. Such equilibrium may be broken with a new vehicle. It is assumed that the time for the vehicles to achieve a new equilibrium is negligible. The assumption is true when the flow rate is low. Every event then leads to one equilibrium. The actual passing time $\bar t_j^{(i)}$ when $i$ vehicles are considered lies at the $i$th equilibrium such that 
\begin{equation}
\bar t_j^{(i)} = f(t_j^*,\bar t_{-j}^{(i)}),\forall j\leq i\text{.}\label{eq: i eq}
\end{equation}

The average delay of the vehicles satisfies
\begin{equation}
\bar d = \lim_{N\rightarrow\infty} \frac{1}{N}\sum_{i} (\bar t_i^{(N)}-t_i^*) = \lim_{N\rightarrow\infty} \frac{1}{N}\sum_{i}d_{i}\text{,}
\end{equation}
where the second equality is due to \eqref{eq: delay}. According to the central limit theorem, the system is ergodic such that the average delay of all vehicles equals the expected delay introduced by any event in the steady state,
\begin{equation}
E(\bar d) = \lim_{i\rightarrow\infty} E(d_{i})\text{.}\label{eq: ergodicity}
\end{equation}

\section{Case Studies\label{sec: policy}}

To illustrate the effectiveness of the model, this section derives traffic properties under two frequently used policies by analysis, event-driven simulation (EDS) as well as conventional time-driven traffic simulation. The two policies are first-in-first-out (FIFO) policy~\cite{dresner2004multiagent} and flexible order (FO) policy~\cite{ahmane2013modeling}. For simplicity, we only consider a two-lane intersection (which is equivalent to lane merging). More detailed analyses are discussed in the extended version~\cite{liu2018analyzing}.

A policy specifies 1) the passing order, and 2) the temporal gap between two consecutive vehicles. The temporal gap refers to the time distance or headway maintained between vehicles. Denote $\Delta_d$ and $\Delta_s$ to be the temporal gap between vehicles from different lanes and the temporal gap between vehicles from the same lane respectively. The gap may be affected by vehicle speed, uncertainties in perception, and etc. When the traffic flow rate is low, we assume $\Delta_s=0$. 

\begin{figure}[t]
\begin{center}
\subfloat[Domain]{
\includegraphics[width=4.3cm]{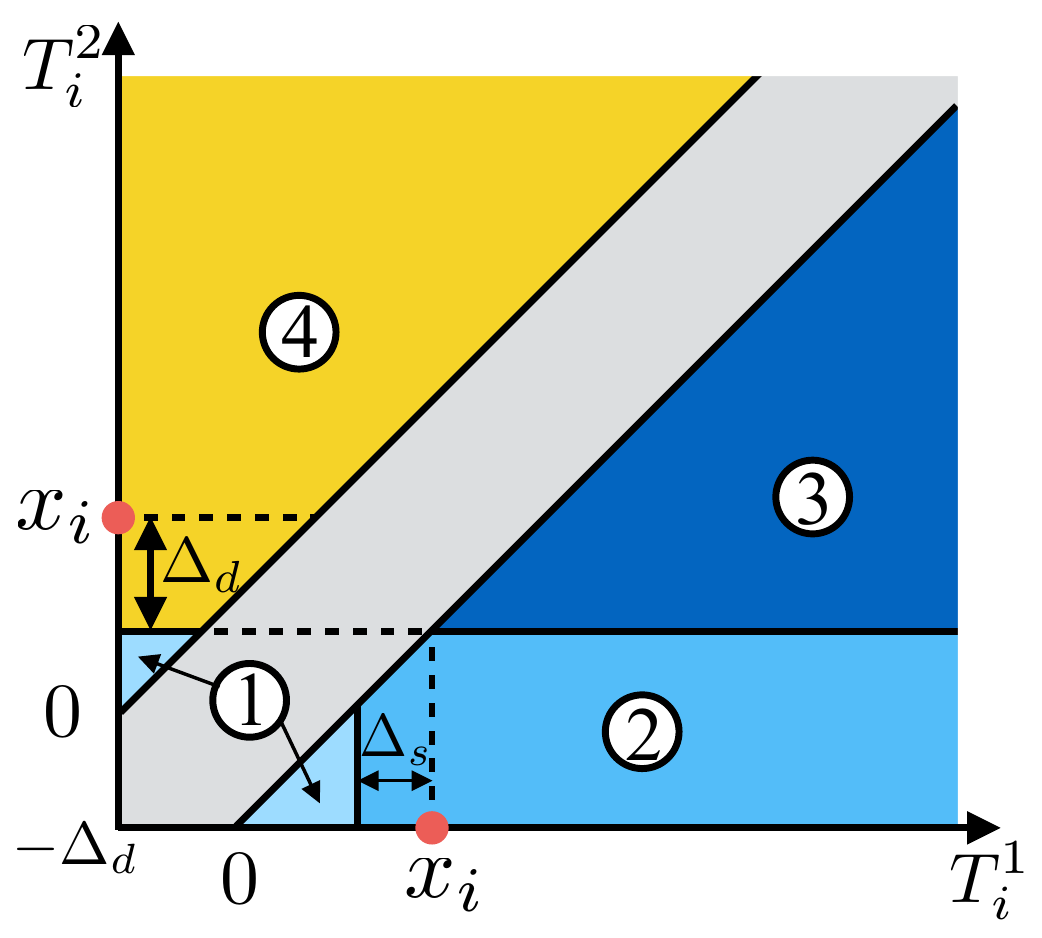}}
\subfloat[Value]{
\includegraphics[width=4.3cm]{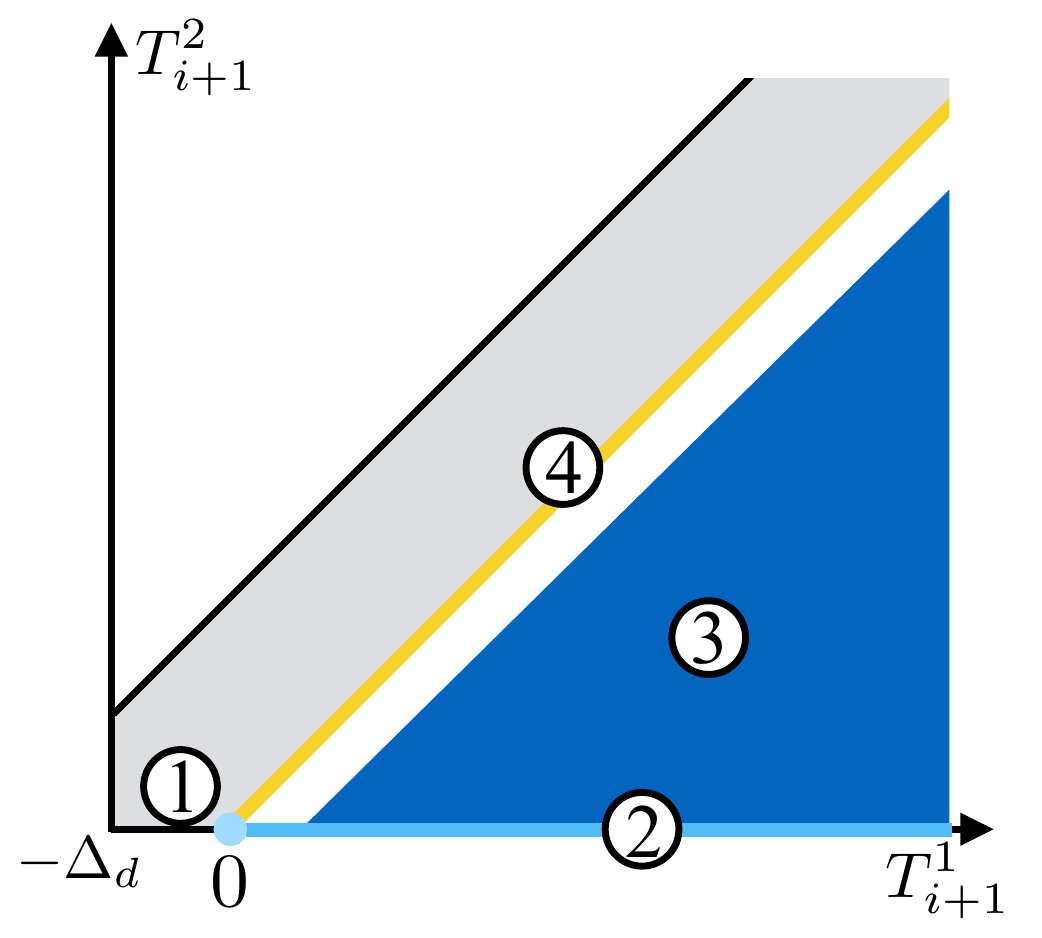}}
\caption{Illustration of the mapping \eqref{eq: dynamic} under FIFO for $s_{i+1} = 1$.}
\label{fig: FIFO mapping}
\end{center}
\vspace{-10pt}
\end{figure}

\subsection{Case 1: Lane Merging with FIFO}
Under FIFO, the passing order is determined by the desired arrival times $\{t_i^*\}_i$ such that the actual passing time for vehicle $i$ should be after the actual passing times for all conflicting vehicles $j$ such that $j<i$.
As the passing order is fixed, the actual passing time will not be affected by later vehicles, i.e., $\bar t_j^{(i)} = \bar t_j^{(j)}$ for all $j<i$. For vehicle $i$,
\begin{equation}
\bar t_i^{(i)} := \max \{t_i^*, \mathcal{D}_i, \mathcal{S}_i\}\text{,}\label{eq: fifo micro}
\end{equation}
where
\begin{eqnarray}
\mathcal{D}_i &=& \max_{j}(\bar t_j^{(i)}+\Delta_d)\text{ s.t. } j<i, (s_{j},s_i)\in \mathcal{G}\text{,}\\
\mathcal{S}_i &=& \max_{j}(\bar t_j^{(i)}+\Delta_s)\text{ s.t. } j<i, s_{j}=s_i\text{.}
\end{eqnarray}

The effect of FIFO is illustrated in \cref{fig: time b}. Following from \eqref{eq: defn T} and \eqref{eq: fifo micro}, the equation \eqref{eq: dynamic} for FIFO can be computed, which is listed in \Cref{table: FIFO mapping} and shown in \cref{fig: FIFO mapping}. Only the case for $s_{i+1}=1$ is shown. The case for $s_{i+1}=2$ can be obtained by switching superscripts $1$ and $2$. To bound the domain from below, let $T_i^{j} = \max\{T_i^{j},-\Delta_d\}$ for all $i$ and $j \in\{ 1,2\}$. The mapping is piece-wise smooth with four smooth components. Region~1 corresponds to where there is a sufficient gap in both lanes for the $(i+1)$th vehicle to pass without delay. Regions 2 and 3 correspond to where the last vehicle is from the ego lane and it causes delay for the $(i+1)$th vehicle. Region 4 corresponds to where the last vehicle is from the other lane and delays the $(i+1)$th vehicle.

\begin{table}[t]
\vspace{5pt}
\caption{The mapping \eqref{eq: dynamic} under FIFO for $s_{i+1} = 1$.}
\vspace{-5pt}
\begin{center}
\begin{tabular}{ccc}
\toprule
Region & Domain & Value\\
\midrule
1 & $\begin{array}{c}T_i^1<x_i-\Delta_s \\ T_i^2<x_i-\Delta_d\end{array}$ & $\begin{array}{cc} T_{i+1}^1 = 0 \\ T_{i+1}^2 = -\Delta_d\end{array}$\\
\midrule
2 & $\begin{array}{c}T_i^1\geq x_i-\Delta_s \\ T_i^2<x_i-\Delta_d \\ T_i^2<T_i^1\end{array}$ & $\begin{array}{cc} T_{i+1}^1 = T_i^1+\Delta_s-x_i \\ T_{i+1}^2 = -\Delta_d\end{array}$\\
\midrule
3 & $\begin{array}{c}T_i^2\geq x_i-\Delta_d \\ T_i^2<T_i^1\end{array}$ & $\begin{array}{cc} T_{i+1}^1 = T_i^1+\Delta_s-x_i \\ T_{i+1}^2 = T_i^2-x_i\end{array}$\\
\midrule
4 & $\begin{array}{c}T_i^2\geq x_i-\Delta_d \\ T_i^2>T_i^1\end{array}$ & $\begin{array}{cc} T_{i+1}^1 = T_i^2+\Delta_d-x_i \\ T_{i+1}^2 = T_i^2-x_i\end{array}$\\
\bottomrule
\end{tabular}
\end{center}
\label{table: FIFO mapping}
\vspace{-10pt}
\end{table}%

Given the dynamic equation, the distribution of traffic delay in \eqref{eq: probability} can be computed. The propagation of $p_{\mathbf{T}_i}$ for $\lambda_1=\SI{0.1}{\per\second}$, $\lambda_2=\SI{0.5}{\per\second}$, $\Delta_d=\SI{2}{\second}$, $\Delta_s=\SI{1}{\second}$ is shown in \cref{fig: fifo multiple} by an event-driven simulation of \eqref{eq: dynamic} with $10,000$ particles. At iteration~1, $P_s(1)$ percent of particles are at $(0,-\Delta_d)$, while the others are at $(-\Delta_d,0)$. Every particle corresponds to a traffic scenario. For conventional time-driven traffic simulation, it is computationally expensive to obtain distributions with $10,000$ traffic scenarios. However, with the event-driven simulation under the proposed model, the distributions can be obtained in real time. The distribution approached steady state at iteration 8 with a unique pattern. Theoretical analysis \cite{liu2018analyzing} also verifies this pattern.

\begin{figure*}[t]
\vspace{3pt}
\subfloat[Case 1: Lane Merging with FIFO. Iterations 2, 3, 4, 8.\label{fig: fifo multiple}]{\centering
%
%
\begin{tikzpicture}

\begin{axis}[%
width=3.3cm,
height=3.3cm,
font=\footnotesize,
at={(1.648in,0.642in)},
scale only axis,
axis on top,
xmin=0.5,
xmax=106,
xtick={1,22,43,64,85,106},
xticklabels={{$-2$},{$0$},{$2$},{$4$},{$6$},{$8$}},
xlabel={Delay in lane 1 [\si{\second}]},
xlabel style={at={(0.5,-0.08)}},
ymin=0.5,
ymax=106,
ylabel={Delay in lane 2 [\si{\second}]},
ylabel style={at={(-0.08, 0.5)}},
ytick={1,22,43,64,85,106},
yticklabels={{$-2$},{$0$},{$2$},{$4$},{$6$},{$8$}},
axis background/.style={fill=white}
]
\addplot [forget plot] graphics [xmin=0.5, xmax=106.5, ymin=0.5, ymax=106.5] {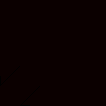};
\end{axis}
\end{tikzpicture}%
\hspace{-2pt}
%
%
\begin{tikzpicture}

\begin{axis}[%
width=3.3cm,
height=3.3cm,
font=\footnotesize,
at={(1.648in,0.642in)},
scale only axis,
axis on top,
xmin=0.5,
xmax=106,
xtick={1,22,43,64,85,106},
xticklabels={{$-2$},{$0$},{$2$},{$4$},{$6$},{$8$}},
xlabel={Delay in lane 1 [\si{\second}]},
xlabel style={at={(0.5,-0.08)}},
ylabel={Delay in lane 2 [\si{\second}]},
ylabel style={at={(-0.08, 0.5)}},
ymin=0.5,
ymax=106,
ytick={1,22,43,64,85,106},
yticklabels={{$-2$},{$0$},{$2$},{$4$},{$6$},{$8$}},
axis background/.style={fill=white}
]
\addplot [forget plot] graphics [xmin=0.5, xmax=106.5, ymin=0.5, ymax=106.5] {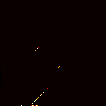};
\end{axis}
\end{tikzpicture}%
\hspace{-2pt}
%
%
\begin{tikzpicture}

\begin{axis}[%
width=3.3cm,
height=3.3cm,
font=\footnotesize,
at={(1.648in,0.642in)},
scale only axis,
axis on top,
xmin=0.5,
xmax=106,
xtick={1,22,43,64,85,106},
xticklabels={{$-2$},{$0$},{$2$},{$4$},{$6$},{$8$}},
xlabel={Delay in lane 1 [\si{\second}]},
xlabel style={at={(0.5,-0.08)}},
ylabel={Delay in lane 2 [\si{\second}]},
ylabel style={at={(-0.08, 0.5)}},
ymin=0.5,
ymax=106,
ytick={1,22,43,64,85,106},
yticklabels={{$-2$},{$0$},{$2$},{$4$},{$6$},{$8$}},
axis background/.style={fill=white}
]
\addplot [forget plot] graphics [xmin=0.5, xmax=106.5, ymin=0.5, ymax=106.5] {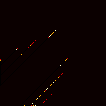};
\end{axis}
\end{tikzpicture}%
\hspace{-2pt}
%
%
\begin{tikzpicture}

\begin{axis}[%
width=3.3cm,
height=3.3cm,
font=\footnotesize,
at={(1.648in,0.642in)},
scale only axis,
axis on top,
xmin=0.5,
xmax=106,
xtick={1,22.6,44.2,65.8,87.4,109},
xticklabels={{$-2$},{$0$},{$2$},{$4$},{$6$},{$8$}},
xlabel={Delay in lane 1 [\si{\second}]},
xlabel style={at={(0.5,-0.08)}},
ylabel={Delay in lane 2 [\si{\second}]},
ylabel style={at={(-0.08, 0.5)}},
ymin=0.5,
ymax=106,
ytick={1,22,43,64,85,106},
yticklabels={{$-2$},{$0$},{$2$},{$4$},{$6$},{$8$}},
axis background/.style={fill=white},
colormap/hot2,
colorbar,
colorbar style={
	at = {(1.02, 1)},
        width=8pt,
        font = \scriptsize,
        ytick={0, 0.5, 1},
	yticklabels={0, 0.1, 0.2}
    }
]
\addplot [forget plot] graphics [xmin=0.5, xmax=109.5, ymin=0.5, ymax=106.5] {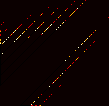};
\end{axis}
\end{tikzpicture}%
}\\
\subfloat[Case 2: Lane Merging with FO. Iterations 2, 3, 4, 8.\label{fig: fo multiple}]{\centering
%
%
\begin{tikzpicture}

\begin{axis}[%
width=3.3cm,
height=3.3cm,
font=\footnotesize,
at={(3.381in,0.801in)},
scale only axis,
axis on top,
xmin=0.5,
xmax=106,
xtick={1,22,43,64,85,106},
xticklabels={{$-2$},{$0$},{$2$},{$4$},{$6$},{$8$}},
xlabel={Delay in lane 1 [\si{\second}]},
xlabel style={at={(0.5,-0.08)}},
ylabel={Delay in lane 2 [\si{\second}]},
ylabel style={at={(-0.08, 0.5)}},
ymin=0.5,
ymax=106,
ytick={1,22,43,64,85,106},
yticklabels={{$-2$},{$0$},{$2$},{$4$},{$6$},{$8$}},
axis background/.style={fill=white}
]
\addplot [forget plot] graphics [xmin=0.5, xmax=106.5, ymin=0.5, ymax=106.5] {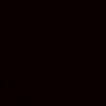};
\end{axis}
\end{tikzpicture}%
\hspace{-2pt}
%
%
\begin{tikzpicture}

\begin{axis}[%
width=3.3cm,
height=3.3cm,
font=\footnotesize,
at={(3.381in,0.801in)},
scale only axis,
axis on top,
xmin=0.5,
xmax=106,
xtick={1,22,43,64,85,106},
xticklabels={{$-2$},{$0$},{$2$},{$4$},{$6$},{$8$}},
xlabel={Delay in lane 1 [\si{\second}]},
xlabel style={at={(0.5,-0.08)}},
ylabel={Delay in lane 2 [\si{\second}]},
ylabel style={at={(-0.08, 0.5)}},
ymin=0.5,
ymax=106,
ytick={1,22,43,64,85,106},
yticklabels={{$-2$},{$0$},{$2$},{$4$},{$6$},{$8$}},
axis background/.style={fill=white}
]
\addplot [forget plot] graphics [xmin=0.5, xmax=106.5, ymin=0.5, ymax=106.5] {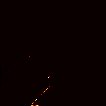};
\end{axis}
\end{tikzpicture}%
\hspace{-2pt}
%
%
\begin{tikzpicture}

\begin{axis}[%
width=3.3cm,
height=3.3cm,
font=\footnotesize,
at={(3.381in,0.801in)},
scale only axis,
axis on top,
xmin=0.5,
xmax=106,
xtick={1,22,43,64,85,106},
xticklabels={{$-2$},{$0$},{$2$},{$4$},{$6$},{$8$}},
xlabel={Delay in lane 1 [\si{\second}]},
xlabel style={at={(0.5,-0.08)}},
ylabel={Delay in lane 2 [\si{\second}]},
ylabel style={at={(-0.08, 0.5)}},
ymin=0.5,
ymax=106,
ytick={1,22,43,64,85,106},
yticklabels={{$-2$},{$0$},{$2$},{$4$},{$6$},{$8$}},
axis background/.style={fill=white}
]
\addplot [forget plot] graphics [xmin=0.5, xmax=106.5, ymin=0.5, ymax=106.5] {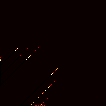};
\end{axis}
\end{tikzpicture}%
\hspace{-2pt}
%
%
\begin{tikzpicture}

\begin{axis}[%
width=3.3cm,
height=3.3cm,
font=\footnotesize,
at={(3.381in,0.801in)},
scale only axis,
axis on top,
xmin=0.5,
xmax=106,
xtick={1,22,43,64,85,106},
xticklabels={{$-2$},{$0$},{$2$},{$4$},{$6$},{$8$}},
xlabel={Delay in lane 1 [\si{\second}]},
xlabel style={at={(0.5,-0.08)}},
ylabel={Delay in lane 2 [\si{\second}]},
ylabel style={at={(-0.08, 0.5)}},
ymin=0.5,
ymax=106,
ytick={1,22,43,64,85,106},
yticklabels={{$-2$},{$0$},{$2$},{$4$},{$6$},{$8$}},
axis background/.style={fill=white},
colormap/hot2,
colorbar,
colorbar style={
	at = {(1.02, 1)},
        width=8pt,
        font = \scriptsize,
        ytick={0, 0.5, 1},
	yticklabels={0, 0.1, 0.2}
    }
]
\addplot [forget plot] graphics [xmin=0.5, xmax=106.5, ymin=0.5, ymax=106.5] {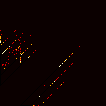};
\end{axis}
\end{tikzpicture}%
}
\caption{Event-driven simulation with $p_{\mathbf{T}_i}$ for $\lambda_1=\SI{0.1}{\per\second}$, $\lambda_2=\SI{0.5}{\per\second}$, $\Delta_d=\SI{2}{\second}$, and $\Delta_s=\SI{1}{\second}$ with $10000$ particles.}
\vspace{-10pt}
\end{figure*}

\subsection{Case 2: Lane Merging with FO}
FO allows high priority vehicles to yield to low priority vehicles if low priority vehicles can arrive earlier. The passing order may change over time. At step $i$, let $\bar t_i^{(i-1)} := \max \{t_i^*, \max_{j<i, s_j=s_i}(\bar t_j^{(i-1)}+\Delta_s)\}$ be the earliest desired time for vehicle $i$ to pass considering its front vehicles in the same lane. Sort the list $\{\bar t_1^{(i-1)},\ldots,\bar t_{i-1}^{(i-1)},\bar t_i^{(i-1)}\}$ in ascending order and record the ranking in an injection $Q$. Ties are broken by index. For the first vehicle in $Q$, i.e., vehicle $k = Q^{-1}(1)$, $\bar t_k^{(i)} : = \bar t_k^{(i-1)}$. By induction, assuming that $\bar t_j^{(i)}$ for $Q(j)<Q(k)$ has been computed, then
\begin{equation}
\bar t_k^{(i)} : = \max \{\bar t_k^{(i-1)}, \mathcal{D}_k^i,\mathcal{S}_k^i\}\text{,}\label{eq: fo micro}
\end{equation}
where 
\small\begin{eqnarray}
\mathcal{D}_k^i &=& \max_{j}(\bar t_j^{(i)}+\Delta_d)\text{ s.t. } Q(j)<Q(k), (s_{j},s_k)\in \mathcal{G}\text{,}\\
\mathcal{S}_k^i &=& \max_{j}(\bar t_j^{(i)}+\Delta_s)\text{ s.t. } Q(j)<Q(k), s_{j}=s_k\text{.}
\end{eqnarray}\normalsize

Under FO, the actual passing time may change at every step. There is a distributed algorithm \cite{liu2017distributed} for this policy where the vehicles do not need to compute the global passing order. \Cref{fig: time c} shows the effect of FO. Vehicles in the same direction tend to form groups and pass together.

Following from \eqref{eq: defn T} and \eqref{eq: fo micro}, the dynamic equation \eqref{eq: dynamic} for FO can be computed, which is listed in \Cref{table: FO mapping} and illustrated in \cref{fig: FO mapping} for $s_{i+1}=1$. There are eight smooth components in the mapping. Regions 1 to 4 are the same as in the FIFO case such that vehicle $i+1$ passes the intersection after all other vehicles. Regions 5 to 8 correspond to where vehicle $i+1$ passes the intersection before the last vehicle in the other lane. In regions 5 and 7, vehicle $i+1$ does not experience delay due to sufficient gap in the ego lane. The last vehicle in the other lane is delayed in region 5, and not delayed in region 7. Regions 6 and 8 correspond to where the $(i+1)$th vehicle is delayed by the last vehicle in the ego lane but can still go before the last vehicle in the other lane. Delay is caused in the other lane in region 6.

Given the dynamic equation, the distribution of delay in \eqref{eq: probability} can be computed. \Cref{fig: fo multiple} shows the event-driven simulation with the same conditions as the FIFO case. FO generates less delay than FIFO. However, the distribution under FO no longer has the ``zebra'' pattern shown in FIFO. We investigate the steady state distribution of delay for $\Delta_s = 0$ and leave the case of $\Delta_s > 0$ for future work. When $\Delta_s = 0$, the mapping $p_{\mathbf{T}_i}\mapsto p_{\mathbf{T}_{i+1}}$ is a contraction as shown in \cref{fig: FO convergence}. \Cref{prop: fo delta_s=0} provides a solution of $p_{\mathbf{T}}$ when $\Delta_s = 0$ and $\lambda_1=\lambda_2$. 
As the problem is symmetric, define $g(t) := p_{\mathbf{T}}(t, t-\Delta_d) = p_{\mathbf{T}}(t-\Delta_d, t)$. The function $g(t)$ represents half of the probability density that $t$ equals the maximum lane delay. Let the finite part of the function be $\tilde g(t)$ and the delta component be $\widehat g(t)$, which is nonzero only at $0$ and $\Delta_d$.

\begin{figure}[t]
\begin{center}
\subfloat[Domain]{
\includegraphics[width=4.3cm]{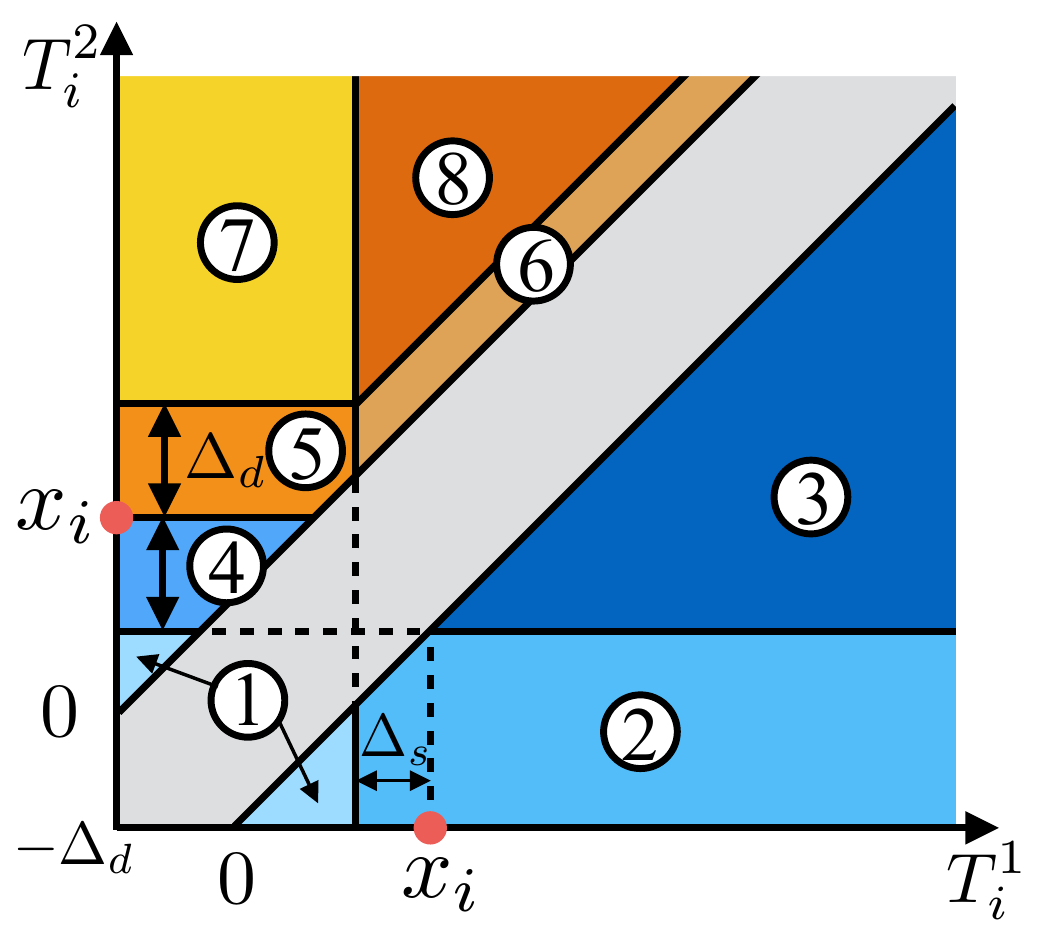}}
\subfloat[Value]{
\includegraphics[width=4.3cm]{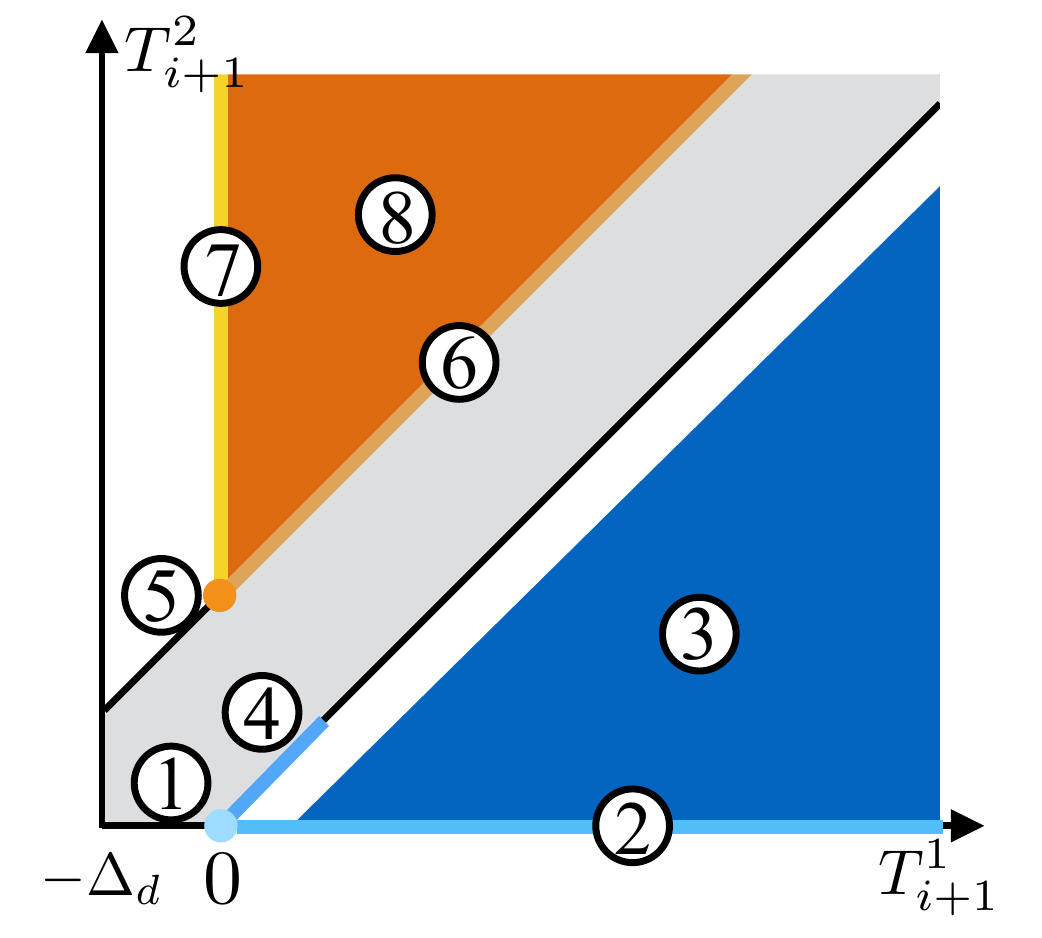}}
\caption{Illustration of the mapping \eqref{eq: dynamic} under FO for $s_{i+1} = 1$.}
\label{fig: FO mapping}
\end{center}
\vspace{-10pt}
\end{figure}

\begin{table}[t]
\vspace{5pt}
\caption{The mapping \eqref{eq: dynamic} under FO for $s_{i+1} = 1$.}
\vspace{-5pt}
\begin{center}
\begin{tabular}{ccc}
\toprule
 & Domain & Value\\
\midrule
1 & $\begin{array}{c}T_i^1<x_i-\Delta_s \\ T_i^2<x_i-\Delta_d\end{array}$ & $\begin{array}{cc} T_{i+1}^1 = 0 \\ T_{i+1}^2 = -\Delta_d\end{array}$\\
\midrule
2 & $\begin{array}{c}T_i^1\geq x_i-\Delta_s \\ T_i^2<x_i-\Delta_d \\ T_i^2<T_i^1\end{array}$ & $\begin{array}{cc} T_{i+1}^1 = T_i^1+\Delta_s-x_i \\ T_{i+1}^2 = -\Delta_d\end{array}$\\
\midrule
3 & $\begin{array}{c}T_i^2\geq x_i-\Delta_d \\ T_i^2<T_i^1\end{array}$ & $\begin{array}{cc} T_{i+1}^1 = T_i^1+\Delta_s-x_i \\ T_{i+1}^2 = T_i^2-x_i\end{array}$\\
\midrule
4 & $\begin{array}{c}T_i^2\in [x_i-\Delta_d, x_i) \\ T_i^2>T_i^1\end{array}$ & $\begin{array}{cc} T_{i+1}^1 = T_i^2+\Delta_d-x_i \\ T_{i+1}^2 = T_i^2-x_i\end{array}$\\
\midrule
5 & $\begin{array}{c}T_i^2\in [x_i, x_i+\Delta_d) \\ T_i^1<x_i-\Delta_s \end{array}$ & $\begin{array}{cc} T_{i+1}^1 = 0 \\ T_{i+1}^2 = \Delta_d \end{array}$\\
\midrule
6 & $\tiny\begin{array}{c}T_i^2-T_i^1\in [\Delta_d,\Delta_d+\Delta_s] \\ T_i^1\geq x_i-\Delta_s \end{array}$ & $\tiny\begin{array}{cc} T_{i+1}^1 = T_i^1-x_i+\Delta_s \\ T_{i+1}^2 = T_i^1-x_i+\Delta_s+\Delta_d \end{array}$\\
\midrule
7 & $\begin{array}{c}T_i^1<x_i-\Delta_s \\ T_i^2\geq x_i+\Delta_d \end{array}$ & $\begin{array}{cc} T_{i+1}^1 = 0 \\ T_{i+1}^2 = T_i^2-x_i \end{array}$\\
\midrule
8 & $\begin{array}{c}T_i^2-T_i^1 > x_i+\Delta_d+\Delta_s \\ T_i^1\geq x_i-\Delta_s \end{array}$ & $\begin{array}{cc} T_{i+1}^1 = T_i^1-x_i+\Delta_s \\ T_{i+1}^2 = T_i^2-x_i \end{array}$\\
\bottomrule
\end{tabular}
\end{center}
\label{table: FO mapping}
\vspace{-10pt}
\end{table}%

\begin{figure*}[t]
\begin{center}
\subfloat[Iteration 1.]{
%
%
\begin{tikzpicture}

\begin{axis}[%
width=3.3cm,
height=3.3cm,
font=\footnotesize,
at={(1.648in,0.642in)},
scale only axis,
axis on top,
xmin=1,
xmax=106,
xtick={1,22,43,64,85,106},
xticklabels={{$-2$},{$0$},{$2$},{$4$},{$6$},{$8$}},
xlabel={Delay in lane 1 [\si{\second}]},
xlabel style={at={(0.5,-0.08)}},
ylabel={Delay in lane 2 [\si{\second}]},
ylabel style={at={(-0.08, 0.5)}},
ymin=1,
ymax=106,
ytick={1,22,43,64,85,106},
yticklabels={{$-2$},{$0$},{$2$},{$4$},{$6$},{$8$}},
axis background/.style={fill=white}
]
\addplot [forget plot] graphics [xmin=0.5, xmax=179.5, ymin=0.5, ymax=127.5] {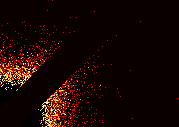};
\end{axis}
\end{tikzpicture}%
}
\subfloat[Iteration 3.]{
%
%
\begin{tikzpicture}

\begin{axis}[%
width=3.3cm,
height=3.3cm,
font =\footnotesize,
at={(1.648in,0.642in)},
scale only axis,
axis on top,
xmin=1,
xmax=106,
xtick={1,22,43,64,85,106},
xticklabels={{$-2$},{$0$},{$2$},{$4$},{$6$},{$8$}},
xlabel={Delay in lane 1 [\si{\second}]},
xlabel style={at={(0.5,-0.08)}},
ylabel={Delay in lane 2 [\si{\second}]},
ylabel style={at={(-0.08, 0.5)}},
ymin=1,
ymax=106,
ytick={1,22,43,64,85,106},
yticklabels={{$-2$},{$0$},{$2$},{$4$},{$6$},{$8$}},
axis background/.style={fill=white}
]
\addplot [forget plot] graphics [xmin=0.5, xmax=162.5, ymin=0.5, ymax=107.5] {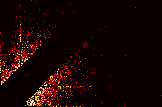};
\end{axis}
\end{tikzpicture}%
}
\subfloat[Iteration 10.]{
%
%
\begin{tikzpicture}

\begin{axis}[%
width=3.3cm,
height=3.3cm,
font =\footnotesize,
at={(1.648in,0.642in)},
scale only axis,
axis on top,
xmin=1,
xmax=106,
xtick={1,22,43,64,85,106},
xticklabels={{$-2$},{$0$},{$2$},{$4$},{$6$},{$8$}},
xlabel={Delay in lane 1 [\si{\second}]},
xlabel style={at={(0.5,-0.08)}},
ylabel={Delay in lane 2 [\si{\second}]},
ylabel style={at={(-0.08, 0.5)}},
ymin=1,
ymax=106,
ytick={1,22,43,64,85,106},
yticklabels={{$-2$},{$0$},{$2$},{$4$},{$6$},{$8$}},
axis background/.style={fill=white}
]
\addplot [forget plot] graphics [xmin=0.5, xmax=130.5, ymin=0.5, ymax=106.5] {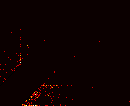};
\end{axis}
\end{tikzpicture}%
}
\subfloat[Iteration 20.]{
%
%
\begin{tikzpicture}

\begin{axis}[%
width=3.3cm,
height=3.3cm,
font =\footnotesize,
at={(1.648in,0.642in)},
scale only axis,
axis on top,
xmin=1,
xmax=106,
xtick={1,22,43,64,85,106},
xticklabels={{$-2$},{$0$},{$2$},{$4$},{$6$},{$8$}},
xlabel={Delay in lane 1 [\si{\second}]},
xlabel style={at={(0.5,-0.08)}},
ylabel={Delay in lane 2 [\si{\second}]},
ylabel style={at={(-0.08, 0.5)}},
ymin=1,
ymax=106,
ytick={1,22,43,64,85,106},
yticklabels={{$-2$},{$0$},{$2$},{$4$},{$6$},{$8$}},
axis background/.style={fill=white},
colormap/hot2,
colorbar,
colorbar style={
	at = {(1.02, 1)},
        width=8pt,
        font = \scriptsize,
        ytick={0, 0.5, 1},
	yticklabels={0, 0.1, 0.2}
    }
]
\addplot [forget plot] graphics [xmin=0.5, xmax=106.5, ymin=0.5, ymax=106.5] {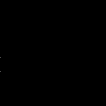};
\end{axis}
\end{tikzpicture}%
}
\caption{Illustration of the convergence of \eqref{eq: probability} under FO for $\lambda_1 = \SI{1.1}{\per\second}$, $\lambda_2=\SI{0.5}{\per\second}$, $\Delta_d = \SI{2}{\second}$, and $\Delta_s = \SI{0}{\second}$ with $10000$ particles.}
\label{fig: FO convergence}
\end{center}
\vspace{-10pt}
\end{figure*}

\begin{prop}[Steady State Distribution for $\Delta_s=0$ under FO]\label{prop: fo delta_s=0}
If $\Delta_s = 0$ and $\lambda_1=\lambda_2$, we have
\begin{equation}\label{eq: fo g}
\tilde g(t) = C e^{\frac{\lambda}{2} t},~~\widehat g(0) = \frac{2}{\lambda}C, ~~\widehat g(\Delta_d) = \frac{1}{2}-\frac{2e^{\frac{\lambda}{2}\Delta_d}}{\lambda}C\text{,}
\end{equation}
where $C = \frac{\lambda(1+e^{-\lambda\Delta_d})}{8\left[e^{\frac{\lambda}{2}\Delta_d} +e^{-\frac{\lambda}{2}\Delta_d} - 1 \right]}$.
\end{prop}
\begin{proof}
For $t\in(0,\Delta_d)$, using the mapping in regions 3, 4, and 8, we get the following steady state relationship $\tilde g(t) = \frac{1}{2}\left[\int_{\Delta_d-t}^\infty g(t+x-\Delta_d)p_xdx + \int_{0}^\infty g(t+x)p_xdx\right]$. 
Multiply both sides by $e^{-\lambda t}$, and then differentiate with respect to $t$ to get
\begin{equation}
\tilde g' -\lambda \tilde g = -\frac{\lambda}{2} \tilde g\text{.}
\end{equation}
Hence, $\tilde g = C e^{\frac{\lambda}{2} t}$ for some constant $C$. Now we solve for the constant $C$.
Due to symmetry, $\int_0^{\Delta_d} g(t)dt = \frac{1}{2}$. Using the fact that $g = \tilde g+\widehat g$ and $\tilde g = C e^{\frac{\lambda}{2} t}$, we get
\begin{equation}
\widehat g(0) + \frac{2C(e^{\frac{\lambda}{2}\Delta_d}-1)}{\lambda} + \widehat g(\Delta_d) = \frac{1}{2}\text{.}\label{eq: fifo sum 1}
\end{equation}

Consider region 1. The point mass at $0$ is $\widehat g(0) = \frac{1}{2}\int_{\Delta_d}^\infty \int_0^{x-\Delta_d} g(\tau)d\tau p_xdx  + \frac{1}{2}\int_{0}^\infty \int_0^{x}g(\tau)d\tau p_xdx$. By changing the order of integration, we get $\widehat g(0) = \frac{1}{2}\int_{0}^\infty \int_{\tau+\Delta_d}^{\infty} p_xdx g(\tau)d\tau + \frac{1}{2}\int_{0}^\infty \int_\tau^{\infty}p_xdx g(\tau)d\tau = \frac{1}{2} \int_0^\infty e^{-\lambda(\tau+\Delta_d)}g(\tau)d\tau +\frac{1}{2}\int_0^\infty e^{-\lambda\tau}g(\tau)d\tau$. Hence,
\begin{equation}
\widehat g(0) = \frac{e^{-\lambda\Delta_d}+1}{2}\mathcal{I}\text{,}
\end{equation}\normalsize
where $\mathcal{I} = \int_0^\infty e^{-\lambda\tau}g(\tau)d\tau$. Plugging in the expression of $g(\tau)$, we have
\begin{equation}
\mathcal{I}=\widehat g(0)+e^{-\lambda\Delta_d}\widehat g(\Delta_d) + \frac{2C(1-e^{-\frac{\lambda \Delta_d}{2}})}{\lambda}\text{.}
\end{equation}

Consider region 5. The point mass at $\Delta_d$ is $\widehat g(\Delta_d) = \frac{1}{2}\int_0^{\infty}\int_x^{\Delta_d} g(\tau)d\tau p_xdx$. By changing the order of integration, $\widehat g(\Delta_d) = \frac{1}{2}\int_0^{\Delta_d}\int_0^{\tau} p_x dx g(\tau)d\tau = \frac{1}{2}\int_0^{\Delta_d}(1-e^{-\lambda\tau})g(\tau)d\tau$. Then,
\begin{equation}
\widehat g(\Delta_d)  = \frac{1}{2}\left[\frac{1}{2} - \mathcal{I}\right]\text{.}\label{eq: fifo point mass at delta_d}
\end{equation}
We combine \eqref{eq: fifo sum 1} to \eqref{eq: fifo point mass at delta_d} to verify \Cref{prop: fo delta_s=0}.
\end{proof}

According to \eqref{eq: fo g}, the probability of zero-delay ($2\widehat g(0)$) and the probability of $\Delta_d$-delay ($2\widehat g(\Delta_d)$) only depend on $\lambda\Delta_d$, i.e., the ratio between the temporal gap and the arrival interval. \Cref{fig: point delay} illustrates those relationships. When the ratio between the temporal gap and the arrival interval increases, the probability of zero-delay decreases while the probability of $\Delta_d$-delay increases. \cref{fig: point delay} also illustrates the result from EDS, which verifies \Cref{prop: fo delta_s=0}.

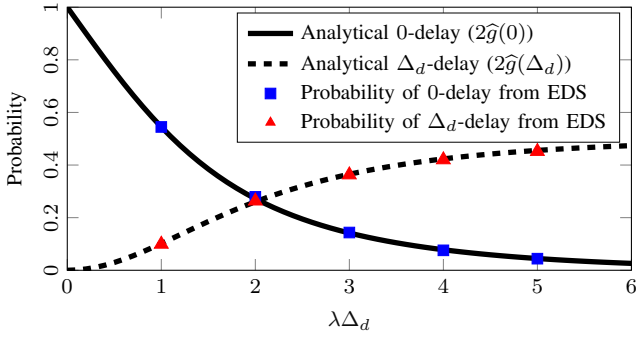
\begin{figure}[t]
\begin{center}
%
%
\begin{tikzpicture}

\begin{axis}[%
width=7.5cm,
height=3.5cm,
font = \footnotesize,
at={(1.011in,2.516in)},
scale only axis,
xmin=0,
xmax=6,
xlabel style={font=\color{white!5!black}\footnotesize},
xlabel={$\lambda\Delta_d$},
ymin=0,
ymax=1,
ylabel={Probability},
yticklabel style={rotate=90},
axis background/.style={fill=white},
legend style={legend cell align=left, align=left, draw=white!6!black},
scatter/classes={%
		a={mark=square*,blue},%
		b={mark=triangle*,red},%
		c={mark=o,draw=black}}
]
\addplot [color=black, line width=2.0pt]
  table[row sep=crcr]{%
0	1\\
0.1	0.950043106396634\\
0.2	0.900354327806761\\
0.3	0.851220706624273\\
0.4	0.802935405882331\\
0.5	0.755782394184816\\
0.6	0.710023038886288\\
0.7	0.66588555944399\\
0.8	0.623557853702323\\
0.9	0.583183768074769\\
1	0.544862512468571\\
1.1	0.508650666508067\\
1.2	0.474566096602553\\
1.3	0.44259308897361\\
1.4	0.412688072023055\\
1.5	0.384785418207774\\
1.6	0.358802950395725\\
1.7	0.334646907891192\\
1.8	0.312216239137689\\
1.9	0.291406175389907\\
2	0.272111101803196\\
2.1	0.254226782201912\\
2.2	0.237652015607469\\
2.3	0.222289811098594\\
2.4	0.208048166978768\\
2.5	0.194840534008335\\
2.6	0.182586033241344\\
2.7	0.171209488631106\\
2.8	0.160641324233801\\
2.9	0.150817366257858\\
3	0.141678581741492\\
3.1	0.133170778421523\\
3.2	0.125244284367924\\
3.3	0.11785362110226\\
3.4	0.11095718005678\\
3.5	0.10451690921462\\
3.6	0.0984980144549617\\
3.7	0.0928686783778344\\
3.8	0.0875997980863778\\
3.9	0.0826647424624717\\
4	0.0780391288045684\\
4.1	0.0737006182393335\\
4.2	0.0696287290196165\\
4.3	0.0658046666396206\\
4.4	0.0622111696023174\\
4.5	0.0588323696398381\\
4.6	0.0556536651963035\\
4.7	0.0526616070204188\\
4.8	0.0498437947718377\\
4.9	0.047188783613197\\
5	0.044685999833375\\
5.1	0.0423256646230368\\
5.2	0.0400987251982152\\
5.3	0.0379967925397438\\
5.4	0.0360120850846592\\
5.5	0.0341373777695656\\
5.6	0.0323659558850783\\
5.7	0.030691573254754\\
5.8	0.0291084143014505\\
5.9	0.0276110596090289\\
6	0.0261944546279714\\
6.1	0.0248538812101305\\
6.2	0.0235849316907653\\
6.3	0.0223834852655627\\
6.4	0.0212456864368033\\
6.5	0.0201679253264937\\
6.6	0.0191468196754312\\
6.7	0.0181791983660467\\
6.8	0.0172620863237177\\
6.9	0.0163926906662762\\
7	0.0155683879848437\\
7.1	0.0147867126510844\\
7.2	0.0140453460566401\\
7.3	0.01334210670003\\
7.4	0.0126749410447999\\
7.5	0.012041915080292\\
7.6	0.0114412065231888\\
7.7	0.0108710976040476\\
7.8	0.0103299683884592\\
7.9	0.00981629058731967\\
8	0.00932862181504897\\
8.1	0.00886560025848372\\
8.2	0.00842593972266551\\
8.3	0.00800842502288188\\
8.4	0.00761190769513338\\
8.5	0.00723530199973321\\
8.6	0.00687758119502488\\
8.7	0.00653777406025579\\
8.8	0.00621496164849524\\
8.9	0.00590827425215433\\
9	0.00561688856517314\\
9.1	0.0053400250273032\\
9.2	0.00507694533714574\\
9.3	0.0048269501217231\\
9.4	0.00458937675137208\\
9.5	0.00436359728966644\\
9.6	0.00414901656890907\\
9.7	0.00394507038249244\\
9.8	0.00375122378611513\\
9.9	0.0035669695004702\\
10	0.00339182640859347\\
};
\addlegendentry{Analytical $0$-delay ($2\widehat g(0)$)}

\addplot [color=black, dashed, line width=2.0pt]
  table[row sep=crcr]{%
0	0\\
0.1	0.00124714193392705\\
0.2	0.00495458094441947\\
0.3	0.0110226349245539\\
0.4	0.0192924806308653\\
0.5	0.0295561963815835\\
0.6	0.0415691473774646\\
0.7	0.0550634115161676\\
0.8	0.0697609934386919\\
0.9	0.0853857901482672\\
1	0.101673586085953\\
1.1	0.118379697234631\\
1.2	0.135284193352549\\
1.3	0.152194867431669\\
1.4	0.168948277622802\\
1.5	0.185409263261827\\
1.6	0.201469348630461\\
1.7	0.217044415445607\\
1.8	0.232071966863217\\
1.9	0.24650823751003\\
2	0.260325336646403\\
2.1	0.273508551599106\\
2.2	0.286053889189655\\
2.3	0.297965894692143\\
2.4	0.309255760069474\\
2.5	0.319939714295905\\
2.6	0.330037676673721\\
2.7	0.339572147456797\\
2.8	0.348567307293205\\
2.9	0.357048296766592\\
3	0.365040648708158\\
3.1	0.372569848274597\\
3.2	0.379660998555509\\
3.3	0.386338572360434\\
3.4	0.392626233634783\\
3.5	0.398546714545697\\
3.6	0.404121736602831\\
3.7	0.409371966214418\\
3.8	0.414316996829501\\
3.9	0.418975351301413\\
4	0.423364499351369\\
4.1	0.427500886043188\\
4.2	0.431399968029682\\
4.3	0.435076255025361\\
4.4	0.4385433545235\\
4.5	0.441814018230217\\
4.6	0.444900189052897\\
4.7	0.447813047771195\\
4.8	0.450563058749531\\
4.9	0.453160014231866\\
5	0.455613076901912\\
5.1	0.457930820502578\\
5.2	0.460121268393709\\
5.3	0.462191929992139\\
5.4	0.464149835086976\\
5.5	0.466001566059216\\
5.6	0.467753288060977\\
5.7	0.469410777228012\\
5.8	0.470979447011463\\
5.9	0.472464372722448\\
6	0.473870314387117\\
6.1	0.475201738011183\\
6.2	0.476462835352262\\
6.3	0.477657542296249\\
6.4	0.478789555930807\\
6.5	0.479862350405162\\
6.6	0.480879191661085\\
6.7	0.48184315111536\\
6.8	0.48275711836932\\
6.9	0.483623813016349\\
7	0.484445795613625\\
7.1	0.485225477879856\\
7.2	0.485965132176499\\
7.3	0.486666900325794\\
7.4	0.487332801815066\\
7.5	0.487964741433112\\
7.6	0.48856451638098\\
7.7	0.489133822896319\\
7.8	0.489674262427408\\
7.9	0.49018734739025\\
8	0.490674506539494\\
8.1	0.491137089981584\\
8.2	0.491576373856327\\
8.3	0.491993564711013\\
8.4	0.49238980358937\\
8.5	0.492766169855888\\
8.6	0.493123684774447\\
8.7	0.49346331485874\\
8.8	0.4937859750106\\
8.9	0.494092531461138\\
9	0.49438380452841\\
9.1	0.49466057120433\\
9.2	0.494923567582549\\
9.3	0.49517349113814\\
9.4	0.495411002869127\\
9.5	0.495636729309129\\
9.6	0.49585126441971\\
9.7	0.496055171370393\\
9.8	0.49624898421371\\
9.9	0.496433209462111\\
10	0.496608327573097\\
};
\addlegendentry{Analytical $\Delta_d$-delay ($2\widehat g(\Delta_d)$)}

\addplot[scatter,only marks,mark size=2.0pt,%
		scatter src=explicit symbolic]%
	table[meta=label] {
x     y      label
1.0 0.5448 a
2.0 0.279 a
3.0 0.1435 a
4.0 0.0758 a
5.0 0.0441 a
	};
\addlegendentry{Probability of $0$-delay from EDS}

\addplot[scatter,only marks, mark size=3.0pt,%
		scatter src=explicit symbolic]%
	table[meta=label] {
x     y      label
1.0 0.0991 b
2.0 0.2633 b
3.0 0.3627 b
4.0 0.42 b
5.0 0.4518 b
	};
\addlegendentry{Probability of $\Delta_d$-delay from EDS}

\end{axis}
\end{tikzpicture}%
\caption{The probability of delay with respect to parameter $\lambda\Delta_d$.}
\label{fig: point delay}
\end{center}
\end{figure}

We validate the event-driven model against the time-driven traffic simulation. By ergodicity \eqref{eq: ergodicity}, the mean delay of all vehicles in the time-driven traffic simulation should equal the expectation of the delay induced by any event in the steady state. The statistical mean delay under each scenario is obtained through the time-driven traffic simulation for \SI{10}{\minute}. The details of the simulation is provided in an earlier journal \cite{liu2017distributed}. The expected steady state delay is computed using the result from \Cref{prop: fo delta_s=0}. 
Given \eqref{eq: delay macro}, the steady state probability density of delay for $t\in[0,\Delta_d)$ satisfies
\begin{equation}
p_d(t) = \int_0^{\infty} [g(t+x)+g(t+x-\Delta_d)+g(x-t+\Delta_d)]p_{x}dx\text{.}\label{eq: pd}
\end{equation}
Applying \eqref{eq: fo g}, the distribution of steady state delay becomes
\begin{equation}
P_d(t)  = \frac{4C}{\lambda}\left[e^{\frac{\lambda}{2}t} - e^{\frac{\lambda}{2}(\Delta_d -t)} + e^{\frac{\lambda}{2}\Delta_d-\lambda t}\right] + \frac{1-e^{-\lambda t}}{2}\text{.}\label{eq: Pd}
\end{equation}
\Cref{fig: pd} shows the distributions. The shaded area is obtained though EDS with $10^4$ particles, which validates \eqref{eq: Pd}. 

\begin{figure}[t]
\begin{center}
%
%
\definecolor{mycolor1}{rgb}{0.46600,0.67400,0.18800}%
\begin{tikzpicture}

\begin{axis}[%
width=7.5cm,
height=1.5cm,
font=\footnotesize,
at={(1.011in,2.416in)},
scale only axis,
xmin=0,
xmax=4,
xlabel style={font=\color{white!5!black}\footnotesize},
xlabel={Delay $d$ [\si{\second}]},
ymin=0,
ymax=1,
ylabel={Probablity},
yticklabel style={rotate=90},
ytick={0,0.5,1},
axis background/.style={fill=white},
legend style={at={(0.5,1.1)},anchor=south},
legend style={legend cell align=left, align=left, draw=white!6!black, font=\footnotesize},
legend columns=2
]
\addplot [color=black!25!red, line width=2.0pt]
  table[row sep=crcr]{%
0	0.544862512468571\\
0.1	0.57870433703729\\
0.2	0.615448935045106\\
0.3	0.654930679882148\\
0.4	0.697015294285577\\
0.5	0.741597182612254\\
0.6	0.788597056746132\\
0.7	0.837959828585836\\
0.8	0.88965274473412\\
0.9	0.943663741434946\\
1	1\\
1	1\\
2	1\\
};
\addlegendentry{$\Delta_d=\SI{1}{\second}$ Analytical}

\addplot [color=black!50!red, line width=2.0pt]
  table[row sep=crcr]{%
0	0.272111101803196\\
0.2	0.32767298124413\\
0.4	0.387421587480993\\
0.6	0.450883343544667\\
0.8	0.517816075991622\\
1	0.588171381217475\\
1.2	0.662065313562302\\
1.4	0.73975594489743\\
1.6	0.821626623874399\\
1.8	0.908173991764305\\
2	1\\
2	1\\
4	1\\
};
\addlegendentry{$\Delta_d=\SI{2}{\second}$ Analytical}


\addplot [color=black!75!red, line width=2.0pt]
  table[row sep=crcr]{%
0	0.141678581741492\\
0.3	0.218072795854256\\
0.6	0.294923886968054\\
0.9	0.372198666182505\\
1.2	0.450331102420687\\
1.5	0.530113501667277\\
1.8	0.61262650375549\\
2.1	0.699198351469569\\
2.4	0.791386611547074\\
2.7	0.890977563036718\\
3	1\\
3	1\\
6	1\\
};
\addlegendentry{$\Delta_d=\SI{3}{\second}$ Analytical}

\addplot [color=black, line width=2.0pt]
  table[row sep=crcr]{%
0	0.0780391288045684\\
0.4	0.174578301791518\\
0.8	0.264324914513543\\
1.2	0.348814460442389\\
1.6	0.430052683734533\\
2	0.510371487186262\\
2.4	0.592371972769355\\
2.8	0.678927918373731\\
3.2	0.773233467516934\\
3.6	0.878885977680569\\
4	1\\
4	1\\
8	1\\
};
\addlegendentry{$\Delta_d=\SI{4}{\second}$ Analytical}

\addplot[ybar interval, fill=mycolor1, fill opacity=0.5, draw=none, area legend] table[row sep=crcr] {%
x	y\\
0	0.0866\\
0.04	0.0965\\
0.08	0.1081\\
0.12	0.1172\\
0.16	0.1273\\
0.2	0.1381\\
0.24	0.1472\\
0.28	0.1579\\
0.32	0.1683\\
0.36	0.1789\\
0.4	0.1874\\
0.44	0.1966\\
0.48	0.2058\\
0.52	0.2152\\
0.56	0.2228\\
0.6	0.2309\\
0.64	0.2416\\
0.68	0.2514\\
0.72	0.2595\\
0.76	0.269\\
0.8	0.2776\\
0.84	0.2858\\
0.88	0.2941\\
0.92	0.3022\\
0.96	0.3111\\
1	0.319\\
1.04	0.3258\\
1.08	0.3332\\
1.12	0.3417\\
1.16	0.3497\\
1.2	0.3579\\
1.24	0.3672\\
1.28	0.3758\\
1.32	0.3842\\
1.36	0.3919\\
1.4	0.3974\\
1.44	0.4054\\
1.48	0.4122\\
1.52	0.4212\\
1.56	0.4311\\
1.6	0.4385\\
1.64	0.4473\\
1.68	0.4534\\
1.72	0.4605\\
1.76	0.4681\\
1.8	0.4765\\
1.84	0.4856\\
1.88	0.4935\\
1.92	0.5023\\
1.96	0.5122\\
2	0.5199\\
2.04	0.5274\\
2.08	0.5349\\
2.12	0.5432\\
2.16	0.551\\
2.2	0.5598\\
2.24	0.5658\\
2.28	0.5733\\
2.32	0.5819\\
2.36	0.5919\\
2.4	0.6007\\
2.44	0.6086\\
2.48	0.6168\\
2.52	0.6266\\
2.56	0.6345\\
2.6	0.6435\\
2.64	0.6532\\
2.68	0.6625\\
2.72	0.6707\\
2.76	0.6798\\
2.8	0.6879\\
2.84	0.6979\\
2.88	0.7053\\
2.92	0.7165\\
2.96	0.7262\\
3	0.7333\\
3.04	0.7414\\
3.08	0.7511\\
3.12	0.7622\\
3.16	0.7741\\
3.2	0.7852\\
3.24	0.7947\\
3.28	0.8048\\
3.32	0.8151\\
3.36	0.8257\\
3.4	0.8372\\
3.44	0.8492\\
3.48	0.8585\\
3.52	0.8691\\
3.56	0.8793\\
3.6	0.8885\\
3.64	0.9001\\
3.68	0.9115\\
3.72	0.9235\\
3.76	0.9365\\
3.8	0.9496\\
3.84	0.9615\\
3.88	0.9752\\
3.92	0.9881\\
3.96	1\\
4	1\\
};
\addlegendentry{$\Delta_d = \SI{4}{\second}$ EDS}

\end{axis}
\end{tikzpicture}%
\caption{Illustration of $P_d$ in \eqref{eq: Pd} for $\lambda=\SI{1}{\per\second}$ with different $\Delta_d$.}
\label{fig: pd}
\end{center}
\vspace{-10pt}
\end{figure}
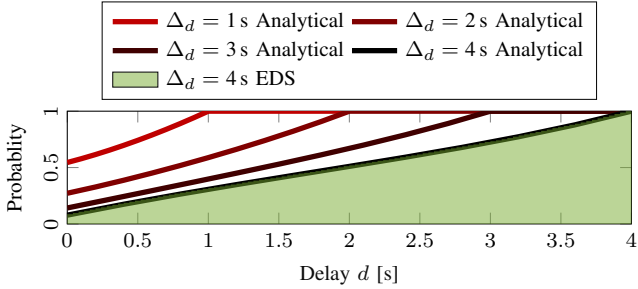

The expected delay can be computed as:
\begin{equation}
E(d) = \int_0^{\Delta_d} t dP_d(t) = \frac{\Delta_d}{2} + \frac{e^{-\lambda\Delta_d}-1}{2\lambda(e^{\frac{\lambda}{2}\Delta_d} +e^{-\frac{\lambda}{2}\Delta_d} - 1)}\text{.}\label{eq: theoretical expected delay}
\end{equation}
When $\lambda\rightarrow 0$, i.e., the traffic flow rate is low, $E(d) \rightarrow \frac{\Delta_d}{2} + \frac{-\lambda\Delta_d+\frac{(\lambda\Delta_d)^2}{2}}{2\lambda}=\frac{\lambda \Delta_d^2}{4}$. 

\begin{figure}[t]
\vspace{3pt}
%
%
\definecolor{mycolor1}{rgb}{0,0,1}%
\definecolor{mycolor2}{rgb}{1,0,0}%
\begin{tikzpicture}

\begin{axis}[%
width=7.5cm,
height=2.5cm,
font=\footnotesize,
at={(1.011in,0.809in)},
scale only axis,
xmin=0,
xmax=1.3,
xlabel style={font=\color{white!5!black}\footnotesize},
xlabel={Traffic flow rate $\lambda$ [\si{\per\second}]},
ylabel style={font=\color{white!5!black}\footnotesize},
ylabel={Expected delay $d$ [\si\second]},
yticklabel style={rotate=90},
ymin=0,
ymax=1.1,
axis background/.style={fill=white},
legend style={legend cell align=left, align=left, draw=white!15!black},
legend pos=north west,
]
\addplot [color=black, line width=1pt]
  table[row sep=crcr]{%
 0 0\\
0.1	0.0574373733355666\\
0.2	0.116329882366903\\
0.3	0.175265794269032\\
0.4	0.232903334103045\\
0.5	0.288087748218657\\
0.6	0.339925580618839\\
0.7	0.387811987893813\\
0.8	0.431418104134563\\
0.9	0.470651540469841\\
1	0.505603944731428\\
1.1	0.536496902971836\\
1.2	0.56363347738677\\
1.3	0.587358921160895\\
1.4	0.608031346714374\\
1.5	0.626001470657415\\
1.6	0.641599804448573\\
1.7	0.655129498817955\\
1.8	0.666863213080389\\
1.9	0.67704267794858\\
2	0.685879939578904\\
2.1	0.693559557291352\\
2.2	0.700241257192331\\
2.3	0.706062717537616\\
2.4	0.71114228652526\\
2.5	0.715581519384431\\
2.6	0.719467478871268\\
2.7	0.722874779797498\\
2.8	0.72586738038535\\
2.9	0.728500135784418\\
3	0.730820135366757\\
3.1	0.732867847763079\\
3.2	0.734678097584656\\
3.3	0.736280896424993\\
3.4	0.737702148714174\\
3.5	0.738964250718415\\
3.6	0.740086598682948\\
3.7	0.741086019945682\\
3.8	0.741977138871257\\
3.9	0.742772687697456\\
4	0.743483770850551\\
};
\addlegendentry{Analytical}

\addplot [color=black, dashed, line width=1pt]
  table[row sep=crcr]{%
0	0\\
1.33333333333333	0.75\\
};
\addlegendentry{Approximated}

\addplot [only marks, color=mycolor1, draw=none, mark=square*, mark size=2pt, mark options={solid, mycolor1}]
  table[row sep=crcr]{%
0.2	0.116\\
0.4	0.233\\
0.6	0.339\\
0.8	0.431\\
1	0.505\\
1.2	0.563\\
};
\addlegendentry{EDS}

\addplot [color=mycolor2, draw=none, mark size=1pt, mark=*, mark options={solid, mycolor2}]
 plot [error bars/.cd, y dir = both, y explicit]
 table[row sep=crcr, y error plus index=2, y error minus index=3]{%
0.2	0.07	0.01	0.01\\
0.4	0.18	0.02	0.02\\
0.6	0.28	0.02	0.02\\
0.8	0.49	0.025	0.025\\
1	0.56	0.03	0.03\\
1.2	1.03	0.06	0.06\\
};
\addlegendentry{Traffic Sim}

\end{axis}
\end{tikzpicture}%
\caption{The expected delay for $\Delta_d = \SI{1.5}{\second}$ and different $\lambda$.}
\label{fig: mean delay}
\vspace{-10pt}
\end{figure}
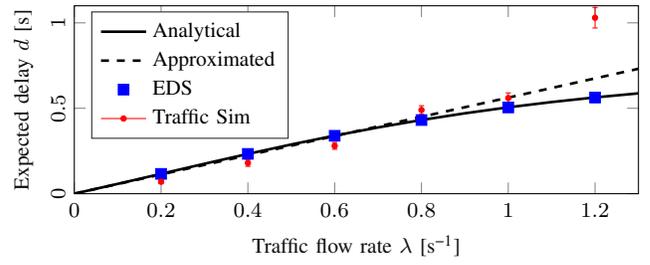

When $\Delta_d = \SI{1.5}{\second}$, \cref{fig: mean delay} shows the analytical expected delay \eqref{eq: theoretical expected delay}, the approximation function $\frac{\lambda \Delta_d^2}{4}$, the expected delay obtained through EDS, and the statistical mean delay in the simulation. When the flow rate is low, the delays obtained through the four methods align well. As the flow increases, the model underestimates the traffic delay because the assumption $\Delta_s=0$ is only valid for small $\lambda$.

\section{Discussion and Conclusion\label{sec: discussion}}
This paper introduced an analytical traffic model for unmanaged intersections. The macroscopic property, i.e., delay at the intersection, was modeled as an event-driven stochastic dynamic process. The macroscopic dynamics encoded the equilibrium resulted from microscopic vehicle interactions. Both the vehicle policies and the road topology could affect the macroscopic dynamics. With the model, the distribution of delay can be obtained through either direct analysis or event-driven simulation, which are more efficient than conventional time-driven traffic simulation and capture more microscopic details than conventional macroscopic flow models. 
The steady state traffic properties were studied, and the accuracy was verified in simulation. 

The potential applications of the analytical model include 1) efficient verification or comparison of policies through analysis or event-driven simulation; 2) policy optimization (e.g., choosing optimal $\Delta_d$) with respect to macroscopic objectives; 3) real-time traffic prediction for intersections; and 4) infrastructure optimization (e.g., designing better road structure and network) to improve traffic efficiency.


\bibliographystyle{IEEEtran}
\bibliography{traffic}

\end{document}